\documentclass[pra,aps,twocolumn,   showpacs,superscriptaddress]{revtex4-1}

\usepackage{graphicx}
\usepackage{dcolumn}
\usepackage{bm}
\usepackage{graphicx}
\usepackage{epsfig}
\usepackage{epsf}
\usepackage{amssymb}
\usepackage{amsmath}
\usepackage{amsthm}
\usepackage{mathtools}
\usepackage{multirow}
\usepackage{cases}
\usepackage[colorlinks=true,linkcolor=blue,citecolor=blue,pdfauthor={ },pdftitle={ },pdfsubject={ },pdfkeywords={ }]{hyperref}
\usepackage{pgfplots}
\usepackage{lipsum}

\usepackage{times}

\theoremstyle{definition}

\newtheorem*{proposition}{Proposition}

\newtheorem{example}{Example}

\begin{document}

\title{Optimal experiment design for quantum state tomography}

\author{Jun Li}
\email{lijunwu@mail.ustc.edu.cn}
\affiliation{Beijing Computational Science Research Center, Beijing 100193, China}
\affiliation{Institute for Quantum Computing, University of Waterloo, Waterloo, Ontario, Canada}

\author{Shilin Huang} 
\email{eurekash.thu@gmail.com}
\affiliation{Institute for Quantum Computing, University of Waterloo, Waterloo, Ontario, Canada} 
\affiliation{Institute for Interdisciplinary Information Sciences, Tsinghua University, Beijing, 100084, China}

\author{Zhihuang Luo}
\affiliation{Beijing Computational Science Research Center, Beijing 100193, China}
\affiliation{Institute for Quantum Computing, University of Waterloo, Waterloo, Ontario, Canada}

\author{Keren Li}
\affiliation{Department of Physics, Tsinghua University, Beijing, 100084, China}
\affiliation{Institute for Quantum Computing, University of Waterloo, Waterloo, Ontario, Canada}
  
\author{Dawei Lu}
\affiliation{Institute for Quantum Computing, University of Waterloo, Waterloo, Ontario, Canada} 
\affiliation{Department of Physics, Southern University of Science and Technology, Shenzhen 518055, China}

\author{Bei Zeng} 
\email{zengb@uoguelph.ca}
\affiliation{Institute for Quantum Computing, University of Waterloo, Waterloo, Ontario, Canada} 
\affiliation{Department of Mathematics \& Statistics, University of Guelph, Guelph, Ontario, Canada}

\begin{abstract} 
Quantum state tomography is an indispensable but costly part of many quantum   experiments. Typically, it   requires measurements to be carried   in a number of different    settings on a fixed experimental setup. The collected data is often informationally overcomplete, with the amount of information redundancy depending on the particular set of measurement settings chosen. This raises a question about how should one   optimally take data  so that   the number of measurement settings necessary can be reduced. Here,   we cast this problem in terms of integer programming. For a given     experimental setup, standard integer programming algorithms allow us to  find the minimum set of readout operations that can realize a  target tomographic task.  We apply the method to certain basic and practical state tomographic problems   in nuclear magnetic resonance experimental systems.   The results show that, considerably less readout operations can be found using our technique than it was by using the previous greedy search strategy. Therefore, our method could     be helpful for simplifying     measurement schemes so as to minimize the  experimental effort.     
\end{abstract}

\pacs{03.67.Lx,76.60.-k,03.65.Yz}

\maketitle

\section{Introduction}

The problem of estimating an unknown quantum state is of fundamental importance in quantum physics \cite{PR04} and especially in the field of quantum information processing, such as  quantum computation \cite{NC00}, quantum cryptography \cite{GRTZ02},  and quantum system identification \cite{BY12}. Quantum state tomography aims to   determine the full state of a quantum system via a series   of quantum measurements. It has become an indispensable tool in almost any experimental physical  setup. 
The standard tomography procedure applied for complete reconstruction of a $d$-dimensional quantum state consists in   projecting the density operator with respect to at least $(d^2-1)$ measurement operators.  A reconstruction based
on Linear-squares inversion \cite{OWV97} or maximum-likelihood estimation \cite{TZERZ11}  is then used to calculate the
best-fit density matrix for the experimentally acquisited data set. Apparently, tomography is not an efficient process and can be extremely computationally costly for even modest-sized  systems.

In recent years, state tomography has been an increasingly   challenging task as the number of controllable qubits in quantum experiments is steadily growing.   
With the   rapid progress of experimental control techniques, the size
of quantum systems with entanglement or coherence prepared in the laboratory has already grown  to 8-10 qubits in photonic systems \cite{Guo11,Pan12,Pan10}, 12 qubits in nuclear magnetic resonance (NMR) systems   \cite{Ray06}
and to even 14 qubits  in ion traps \cite{Blatt11}. Needless to say, performing state estimation tasks on such systems is    tedious and time-consuming.  And   improved techniques for quantum state tomography would certainly impact a wide range of applications in experimental physics.
For example, Ref. \cite{H05} used hundreds of thousands of measurements and weeks of post-processing to get a maximum likelihood
estimate of an entangled state of 8 trapped-ion qubits. Later, this experiment was simplified as there was   put forward a much more economic tomographic scheme, which is based on the concept of mutually unbiased bases and promises to reduce about 95\% of the number of measurements required  \cite{KMFS08}. To give another example, in Ref. \cite{Guo16}  it was shown that, in reconstructing a 14-qubit state, using both smart choice of the state representation and parallel graphic processing unit  programming can  speed up the post-processing   by a factor of $10^4$. Besides of these technical improvements, there also exist various theoretical approaches that are devoted   to enhance the capability  of quantum state tomography \cite{GLFBE10,TWGKSW10,Cramer10}. Most of them either extract partial information or exploit some prior information about the state to be reconstructed.

In this paper, we are concerned with the design of the   measurement scheme in a tomographic experiment. Our study is primarily motivated by a   problem which is present in many experimental platforms, namely   state tomography     often involves informationally overcomplete measurements. The  reason  can be stated as follows. Normally, a  tomographic experiment consists of a series of different measurement settings, and from each single measurement setting a bunch of   data is recorded.  Here,  
a measurement setting  refers to a particular  configuration of the experimental measurement apparatus. For example, in   photonic systems one can tune the waveplates and polarizers   to make   arbitrary local polarization measurements, so   a setting means  the choice of one observable per qubit and repeated projective measurements in the observables' eigenbases \cite{AJK05}.  In NMR, a measurement setting   corresponds to taking a spectrum. Because NMR experiments are performed on a large ensemble of molecules, the expectation values of the  observables  (not necessarily compatible) can be read out from a single spectrum \cite{VC05}.    In both platforms, there could be considerable overlap between the experimental
outcomes acquired from different measurement settings, that is, there are redundant measurements. 
Therefore,   one would ask   about what is the minimum number of measurement settings that suffice to determine the state of the system. A judicious experiment  design would certainly help improve the efficiency of tomographic reconstruction.  The purpose of this paper is to address this question via integer programming techniques.  In the following, we shall first formulate the problem of optimal tomographic experiment design in terms of integer programming. And then, we   concentrate on optimizing the design of   readout pulse set  in the context of nuclear magnetic resonance.

\section{Optimal Experiment Design}

We restrict our consideration to the case of $n$ qubits---higher-dimensional systems can be treated similarly.
An $n$-qubit system's state is represented by a $2^n$-dimensional density matrix denoted as $\rho$, which is Hermitian, semipositive definite and has unit trace. A convenient and equivalent description of the quantum system is given by the Bloch  vector. Let $\left\{ B_k\right\}_{k=1}^{4^n-1}$ be some   orthonormal basis for the   space of traceless Hermitian operators satisfying that  for any $k,j=1,...,n$, there is $\operatorname{Tr} (B_k B_j)/2^n = \delta_{kj}$.  Then decomposed with respect to this basis, $\rho$ can be viewed as a point $\bm{r}$ in a $(4^n-1)$-dimensional real vector space: $\rho = I^{\otimes n}/2^n + \sum_{k=1}^{4^n-1} { \bm{r}_k B_k}$ with   $I$ being the $2$-dimensional identity matrix and $\bm{r}_k= \operatorname{Tr}(\rho B_k)/2^n$. 
Clearly, full tomography amounts to measuring all of the quantities $\left\{ \operatorname{Tr}(\rho B_1), ...,  \operatorname{Tr}(\rho B_{4^n-1})   \right\}$. 

Experimenters primarily work in two different   bases, the computational   basis and the product operator basis. In the computational basis the rows and columns of the density matrix $\rho$ are labeled by the binary expansion of their indices from $\left| 0 \cdots 0\right\rangle$ to $\left| 1 \cdots 1\right\rangle$. The   product operator basis, defined as $\mathcal{P}_n = \left\{ P_k\right\}_{k=1}^{4^n-1} = \left\{I, X, Y,Z\right\}^{\otimes n}/ \left\{ I^{\otimes n} \right\}$ where $X,Y,Z$ are the three Pauli matrices
\[ X =  \left(\begin{array}{cc}
  0 & 1 \\ 
  1 & 0
 \end{array}\right),
 Y =  \left(\begin{array}{cc}
  0 & -i \\ 
  i & 0
 \end{array}\right),
 Z =  \left(\begin{array}{cc}
  1 & 0 \\ 
  0 & -1
 \end{array}\right),
\]
is a  commonly used tool in describing     pulse control experiments. It  provides at the same time physical insight of the experimental setup (e.g., in NMR) and computational convenience \cite{SELBE84}.   In the following we shall work on the product operator basis $\mathcal{P}$, but there is no problem in extending our analysis to other bases.

In performing a tomographic experiment, we send  multiple copies of the state $\rho$ to our measurement apparatus. The apparatus can be configured in  different settings. 
Suppose that under a specific measurement setting, we can read out the information for the following set of operators:
$\mathcal{O} = \left\{ O_1, ...,   O_k, ...\right\}$, here $O_k \in \mathcal{P}$.
The experimental tomography procedure   employs  a series of  measurement settings, each corresponding to observation of a different set of operators.
Here, the switch between measurement settings is implemented through either changing the configuration of the detectors or adding a unitary readout   operation   before data acquisition.  This can be readily seen from the equality $\operatorname{Tr}(U \rho U^\dag E) = \operatorname{Tr}(  \rho  U E U^\dag)$, where $E$ is an arbitrary observable and $U \in SU(2^n)$.    For instance, in order to   measure the three Cartesian components of a spin, and if we can only observe $X$ and $Y$ in one experimental setting, then we will need two readout operations, which can be selected from the set $\left\{I, R_x, R_y \right\}$, here $R_x$, $R_y$ is the $\pi/2$ rotation about $x$, $y$ axis respectively. 

Now suppose   we have the following experimentally available set of readout operations:
$\mathcal{U} = \left\{ U_1, ... , U_j,...\right\}$.
Denote $\mathcal{S} = \left\{S_1,..., S_j,... \right\}$ where $S_j$ corresponds to the set of measurement operators generated through    $U_j$: $S_j = \left\{ U_j O_k U_j^\dag | O_k \in \mathcal{O} \right\}$. We shall assume that $S_j \subseteq \mathcal{P}$ for any $j$. Here is some  abuse of notation as actually we   should  ignore   global phase and coefficient.
Then $\mathcal{S}$ is a collection of  $\left| \mathcal{U} \right|$  subsets of $\mathcal{P}$, each containing $\left| \mathcal{O} \right|$ elements.
Clearly, to ensure   full state tomography, a necessary condition is that    $\mathcal{P}$ should be  covered by $\mathcal{S}$, that is, $
\mathcal{P} = \bigcup_j S_j$.
Now, we can state the central   problem of this work, that is,   to identify the smallest sub-collection of 
$\mathcal{S}$ whose union equals $\mathcal{P}$. More formally, we are considering a standard set cover problem, which we denote by  $\mathbf{P}(\mathcal{P},\mathcal{O},\mathcal{U})$: given $\mathcal{P}$, $\mathcal{O}$, and $\mathcal{U}$, we want to   select a subset of readout operations $\left\{ U_j \right\} \subseteq \mathcal{U}$ with its number of elements as small as possible and such that  
$\mathcal{P}$ is covered by the set $\left\{ U_j O_k U_j^\dag | O_k \in \mathcal{O}, \left\{ U_j \right\} \subseteq \mathcal{U} \right\}$. Fig. \ref{setcover} shows a simple instance of the problem.

\begin{figure}[t]
\begin{center}
\includegraphics[width=0.65\linewidth]{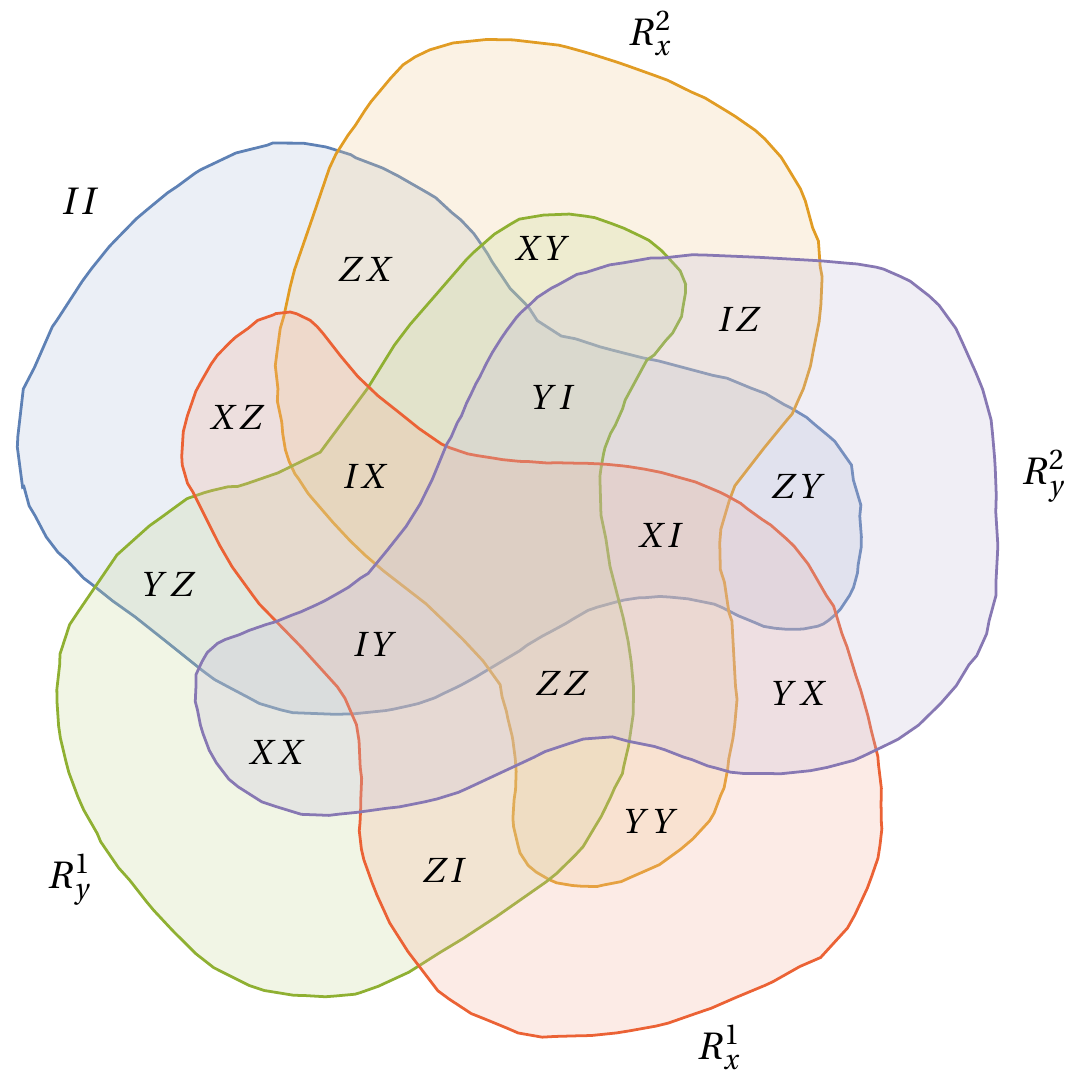}
\end{center}
\caption{Venn diagram visualizing the set cover problem for a   two-qubit   state tomography task. Here the readout operations are restricted to   identity operation and single-qubit rotations. This is a very simple instance occurred when we want to do state tomography for a homonuclear two-spin system in NMR. Clearly, any and no less than four of the five sets suffice to cover the entire measurement basis.}
\label{setcover}
\end{figure}

%

Set cover problem is known to be   NP-hard in general, meaning that   to find an efficient algorithm that can solve it  in reasonable time is unlikely. It is worthwhile to study heuristics for solving the problem with the goal of obtaining a performance guarantee 􏰏or approximation guarantee􏰐 on the heuristic􏰃.  Practically, greedy  strategy is   widely used  for set cover problems. Greedy algorithm proceeds according to a simple rule:  in each step, choose the set $S_j$ containing the largest number of uncovered elements. The algorithm ends until all elements of $\mathcal{P}$ are covered.   
Ref. \cite{LM04} has   exploited this method   in designing readout pulse set in NMR. Note that greedy algorithm generally does not yield the optimal result. Different choices for the first few readout operations in the iteration yield lists that are slightly longer or shorter than those shown. It is in essence an approximation algorithm, which achieves an approximation ratio of $\Theta(\ln |\mathcal{P}|)$ \cite{GT15}. It can even be  shown that no polynomial-time approximation algorithm can achieve a  much better approximation bound \cite{Fei98}.

Here, we attempt to find the optimal solution  for relatively small-sized quantum systems. To this end, we resort to integer linear programming formulation of the   set cover problem \cite{Vazirani04}. Let $x$ denote a $|\mathcal{U}|$-element column vector, in which each  element $x_j$ is a zero-one variable. The intention is that $x_j = 1$ iff set $S_j$ is chosen in the optimal solution. Let $f(x) = \sum_j c_j x_j$ denote the cost function, where $c_j > 0$ is the cost corresponding to  the  choice of $S_j$.  In the case of minimum set cover problem, the cost function is just $f(x) = \sum_j  x_j = \left\| x \right\|_1$. Let $A$ be a $(4^n-1) \times \left| \mathcal{U} \right|$ dimensional matrix with its entries given by $A_{kj} = 1$ ($k=1,...,4^n-1$) if $P_k \in S_j  $, and 0 otherwise. Now we have the following zero-one integer programming problem:
\begin{align}
\min \quad &  \left\| x \right\|_1,  \nonumber \\
\operatorname{s.t.} \quad & Ax \ge 1,  \nonumber \\
& x_j \in \left\{0,1\right\}.  \nonumber
\end{align}
There are a variety of algorithms that can be used to solve integer linear programs exactly, which we do not intend to expand here. Interested readers are referred to  \cite{CBD09, Margot10} for the basics.  
For now, we make several discussions:

(1) \emph{Choose readout operations}. For a quantum system that allows for universal control, readout operation $U$ can be chosen from the   Clifford group $\mathcal{C}$. This is due to that the Clifford group is the normalizer of the Pauli group \cite{Gottesman97}, i.e., for any $U$ a Clifford operation and $O$ a Pauli observable, there is that if ignoring global phase and coefficient,  $U O U^\dag$ gives again a Pauli observable. However,   the size of Clifford group is $| \mathcal{C}(n)| = 8 \cdot \prod_{k=1}^n 2(4^k-1)4^k$, which makes the corresponding integer programming problem soon become  too huge to     handle. Therefore, it is difficult to consider the entire Clifford group in the integer programming approach. More realistically, we would restrict to just    local operations. But notice that in many practical cases   the set of single qubit rotations is not     sufficient for covering the whole measurement basis, and then local Clifford operations should be considered.

(2) \emph{Cost function}. Integer programming allows to  consider  different   cost functions. In practice, more often
than not, technical constraints permit only a nonideal set of measurements. For example, the qubits can be individually addressed whereas nonlocal quantities can not be measured directly. It is also very common that nonlocal operations are less accurate than local operations. In such cases, we would prefer to choose local operations, and this can be achieved simply through assigning low costs to these preferred operations.

(3) \emph{Symmetry Consideration}. Many practical tomographic instances contain a great deal of symmetry. One direct consequence is that the optimal solution is not unique. A very useful technique that allows substantial reduction of the amount of computation required in running integer programming algorithms is to exploit the symmetry of the problem considered \cite{Ostrowski11, Margot10}.  It is desirable to use   professional softwares (e.g., \cite{Gurobi}) to run integer programming algorithms as they have taken into account of the symmetry issue and so can behave much faster. 

\section{NMR Tomography}
We first briefly describe the basics of state tomography experiments in NMR \cite{Lee02}.
We consider weakly coupled liquid-state NMR systems. An NMR sample is placed in a strong static magnetic field. The direction of the static field is, by convention,   defined as the $z$ axis.   The system Hamiltonian takes the following form
\begin{equation}
H = \sum\limits^n_{k =1} {{\Omega _k} Z_k/2}  + \pi \sum\limits^n_{k<j} {{J_{kj}} Z_k Z_j/2},
\label{H}
\end{equation}
where $\Omega_k$ is the precession frequency of the $k$-th spin under the static field, and $J_{kj}$ is the coupling between the $k$th and $j$th spin.
In   NMR   experiment, the sample is wounded with a detection coil. The precessing magnetization
of the sample is detected by the coil and constitutes the free induction decay (FID).  The induced signal is the sum of a number of oscillating waves of different frequencies, amplitudes, and phases. In the spectrometer, this is recorded using two orthogonal detection channels along the $x$- and $y$-axes, known as quadrature detection. The FID is then subjected to Fourier transformation, and the resulting spectral lines are fit, yielding a set of measurement data. 

More precisely, let $\rho$ denote the state at the start of the sampling stage, then   the measured time-domain signal $F(t)$ resulted from the rotating bulk magnetization is essentially a pair of ensemble averages: 
\begin{align}
F(t) & = \operatorname{Tr}\left[ e^{-i Ht} \rho e^{iHt} \cdot \sum_m \left(X_m + i Y_m \right) \right]  \nonumber \\
& = \operatorname{Tr}\left[ \rho \cdot e^{iHt} \sum_m \left(X_m + i Y_m \right) e^{-i Ht} \right].   
\end{align}
From the expression,  the record of the FID signal $F(t)$ can    be thought of  as that the measured operators are $\mathcal{O} = e^{iHt} \sum_k \left(X_k + i Y_k \right) e^{-i Ht}$, which can be calculated easily when the Hamiltonian $H$ is specified. Because $H$  is composed of $Z$- and $ZZ$- terms, the measurement operator set $\mathcal{O}$ consists of only single-quantum coherence operators. When it is desired to measure other operators,    readout pulses should be appended to the experiment.

%


%
 
Now, we give several examples showing how our method developed in the previous section applies to   NMR state tomography.

(1) \emph{Tomography via a single probe qubit}. In certain circumstances, we can do better than the   numerical optimization: we can write down an analytic expression for an optimal readout scheme. Here, we develop a provably optimal scheme, where we intend to perform tomography via   just a single  probe qubit. Of course, to make the scheme work, two conditions must be assumed in advance: (i) the probe qubit is coupled to   each of the other $n-1$ qubits; (ii) the multiplets corresponding to the probe qubit can be well resolved. As we have described,   the set of measurement operators are just single-quantum coherences, so there is
\begin{equation}
\mathcal{O} =   \left\{X, Y \right\} \otimes \left\{ I, Z \right\}^{\otimes n-1}.  \nonumber
\end{equation} 
As such, we would transfer the information of   other qubits     to the probe qubit before detection.
Let $W_{kj}$ denote the SWAP operation between qubit $k$ and $j$. Our candidate   readout operations will be
\begin{equation}
\mathcal{U} =   \left\{ V   W_{1j} : V \in \left\{I, R_x, R_y \right\}^{\otimes n}; j=1,...,n \right\}.   
\label{U}
\end{equation}
Let  $f_n$ denote the number of experiments used. Now we show that the lower bound for the number of experiments necessary to constructing the density operator is  

\begin{proposition}
$f^*_n = (3^n+1)/2$.
\end{proposition}

\begin{proof}
The correctness of the formula  when $n=1$ is trivial. For $n > 1$, the idea to the proof is to find a recursive relation. 

First we show that $(3^n+1)/2$ is a lower bound. Notice that any $U \in \mathcal{U}$ can not change the weight of any Pauli element. So each experiment only gives observation results of two weight-$n$ Pauli elements. Since there are $3^n$ weight-$n$ Pauli elements in $\mathcal{P}$, so at least  $(3^n+1)/2$ readout operations are needed.

Now we show that this lower bound can be achieved. We provide a constructive way of finding an optimal solution. Divide   $\mathcal{P}$ into two parts: $\mathcal{P} = \mathcal{P}^{(1)} \cup \mathcal{P}^{(2)}$, where
\begin{align}
\mathcal{P}^{(1)} & = \left\{X,Y\right\} \otimes \left\{I, X, Y, Z\right\}^{\otimes (n-1)}, \nonumber \\
\mathcal{P}^{(2)} & = \left\{I,Z\right\} \otimes \left\{I, X, Y, Z\right\}^{\otimes (n-1)}.  \nonumber
\end{align}
To cover $\mathcal{P}^{(1)}$, we select readout operations from 
\[\mathcal{U}^{(1)} = I \otimes \left\{I, R_x, R_y\right\}^{\otimes (n-1)};\] 
to cover $\mathcal{P}^{(2)}$, we select readout operations from 
\[\mathcal{U}^{(2)} = \left\{ V   W_{1j} : V \in \left\{I, R_x, R_y \right\}^{\otimes n}; j=2,...,n \right\}.\] 
As $\mathcal{U}^1$ and $\mathcal{U}^2$ has no intersections, we get two seperate subproblems: 

The subproblem $\mathbf{P}(\mathcal{P}^{(1)},\mathcal{O},\mathcal{U}^{(1)})$ is equivalent to problem $\mathbf{P}(\left\{I,X,Y,Z\right\}^{\otimes (n-1)},\left\{ I, Z \right\}^{\otimes n-1},\left\{I, X, Y\right\}^{\otimes (n-1)})$. Since there are $3^{(n-1)}$ weight-$(n-1)$ Pauli elements and each experiment can only measure one of them,  thus there must need $3^{(n-1)}$ experiments.

The subproblem $\mathbf{P}(\mathcal{P}^{(2)},\mathcal{O},\mathcal{U}^{(2)})$ can be reduced to 
\begin{align}
\mathbf{P} & \left( \left\{I,Z\right\} \otimes \left\{I,X,Y,Z\right\}^{\otimes (n-1)}, \right. \nonumber \\
& \left\{I,Z\right\} \otimes \left\{X,Y \right\} \otimes \left\{ I, Z \right\}^{\otimes n-2},  \nonumber \\
& \left. I \otimes \left\{ V   W_{2j} : V \in \left\{I, R_x, R_y \right\}^{\otimes n-1}; j=2,...,n \right\} \right).  \nonumber
\end{align}
That is, any solution to the latter problem can be mapped to a solution to the former by simply changing the probe qubit from 2 to 1.
Note that the latter problem is essentially $\mathbf{P}(\mathcal{P}(n-1),\mathcal{O}(n-1),\mathcal{U}(n-1))$.
Therefore, our construction has a recursive relation: $f_n = f_{n-1} + 3^{n-1}$. Starting with $f^*_1 = 2$, we then get $f_n = (3^n+1)/2$. Together with the   bound analysis at the beginning, we conclude that the  lower bound $(3^n+1)/2$ is exact.
\end{proof}

\begin{table}[b]
{\renewcommand{\arraystretch}{1.5}\footnotesize
\begin{center}
\begin{tabular}{ccl}
\hline
\hline
{\small $n$}  & {\small $f^*_n$}   & {\small A solution instance} \\
\hline
1  & 2   &  $I$, $R_x$
\\
\hline
  2  & 4   &  $II$, $R_x^2$, $R_y^2$, $R_x^1 R_x^2$   
\\
\hline
3  & 7 & $III$, $R_y^3$, $R_y^1$, $R_y^2 R_y^3$, $R_x^1 R_y^2 R_x^3$, $R_x^1 R_x^2 R_y^3$, $R_x^1 R_x^2 R_x^3$  
\\
\hline
4  & 15 &  $IIII$,  $R_x^4$,  $R_x^1R_x^4$,  $R_x^1R_y^4$,  $R_y^1R_y^2$,  $R_x^2R_y^3R_y^4$,  $R_y^2R_x^3R_y^4$, \\
& & $R_y^2R_y^3R_y^4$,  $R_x^1R_x^2R_x^3$,  $R_x^1R_y^2R_y^3$,  $R_y^1R_x^3R_x^4$,  $R_y^1R_y^2R_y^3$, \\
& & $R_x^1R_x^2R_x^3R_x^4$,  $R_y^1R_x^2R_x^3R_x^4$,  $R_y^1R_x^2R_y^3R_y^4$
\\
\hline
  5  & 33 &  
$IIIII$,  $R_x^5$,  $R_x^4R_y^5$,  $R_y^4R_y^5$,  $R_x^3R_y^5$,  $R_y^3R_y^5$,  $R_x^2R_x^3R_x^4$, \\
& & $R_y^2R_y^3R_y^4$,  $R_x^1R_y^3R_x^4$,  $R_x^1R_x^2R_y^4$,  $R_x^1R_y^2R_y^5$,  $R_x^1R_y^2R_x^3$, \\
& &  $R_y^1R_x^3R_y^4$,  $R_y^1R_y^3R_y^5$,  $R_y^1R_x^2R_y^5$,  $R_y^1R_x^2R_y^3$,  $R_y^1R_y^2R_y^5$, \\
& & $R_y^1R_y^2R_x^4$,  $R_x^2R_x^3R_x^4R_x^5$,  $R_y^2R_y^3R_x^4R_x^5$,  $R_y^2R_y^3R_y^4R_x^5$, \\
& & $R_x^1R_y^3R_x^4R_x^5$,  $R_x^1R_x^2R_y^4R_x^5$,  $R_x^1R_y^2R_x^3R_x^5$,  $R_y^1R_x^3R_y^4R_x^5$, \\
& & $R_y^1R_x^2R_y^3R_x^5$,  $R_y^1R_y^2R_x^4R_x^5$,  $R_x^1R_x^2R_x^3R_x^4R_y^5$, \\
& & $R_x^1R_x^2R_x^3R_y^4R_y^5$, $R_x^1R_x^2R_y^3R_x^4R_y^5$,  $R_x^1R_x^2R_y^3R_y^4R_y^5$, \\
& & $R_y^1R_y^2R_x^3R_x^4R_y^5$,  $R_y^1R_y^2R_x^3R_y^4R_y^5$
\\
\hline
\hline
\end{tabular}
\end{center} 
}
\caption{Examples of optimal pulse set that yields complete tomography on $n$-spin homonuclear systems for $n$ between 1 and 5.}
\label{optimal}
\end{table}

\begin{example}
Iodotrifluroethylene (C$_2$F$_3$I)  dissolved in d-chloroform. This molecule was used as a quantum information processor in Ref. \cite{Luo16,Li16}. See the following figure  for the sample's molecular structure:
\begin{center}
  \begin{tikzpicture}[scale=0.01]

  \node at (-50,0) {\footnotesize  C$_1$};
  \node at (50,0) {\footnotesize   C};
  \node at (-115,-65) {\footnotesize  F$_2$};
  \draw  (-25,5) -- (25,5);
  \draw  (-25,-5) -- (25,-5);
  \draw (-65,-25) -- (-90,-50);
  \draw  (-65,25) -- (-90,50) node [anchor=south east] {\footnotesize  F$_1$};
  \draw   (65,-25) -- (90,-50) node [purple,anchor=north west] {\footnotesize I};
  \draw  (65,25) -- (90,50) node [anchor=south west] {\footnotesize  F$_3$};

  \end{tikzpicture} 
\end{center}
Here, the  $^{13}$C (C$_1$) nucleus and the three $^{19}$F nuclei (F$_1$, F$_2$ and F$_3$) constitute a four-qubit system. For this molecule, only C$_1$ can be adequately observed. Therefore  to observe those operators that are relevant to the state of the fluorines it is necessary to   swap the state between them and the carbon before observation. So readout operations can be chosen from the set in Eq. (\ref{U}). According to  our proposition, the minimum number of readout single-qubit rotation pulses is 41.
\end{example}

(2) Now we consider   homonuclear systems, assuming that  all peaks on the spectrum can be adequately resolved and observed. In homonuclear systems, all the multiplets are experimentally observed on the same spectrometer channel, and so  one readout operation, namely one spectrum, contains $n$ well-resolved multiplets and yields $n 2^n$ expansion coefficients.  In such an   ideal case,  we can observe all the single-quantum transition operators from the experimental spectrum. In other words,   we will have the following collection of Pauli measurements that can be accessed in a single experiment
\begin{equation}
\mathcal{O} = \bigcup_{m=1}^{n} \left( \left\{ I,Z \right\}^{\otimes (m-1)} \otimes \left\{ X_m, Y_m \right\} \otimes \left\{ I,Z \right\}^{\otimes n-m} \right).  \nonumber
\end{equation} 
One can choose the following   readout operation set
\begin{equation}
\mathcal{U} = \left\{I, R_x, R_y \right\}^{\otimes n}.  \nonumber
\end{equation} 
Ref. \cite{LM04} referred $\mathcal{U}$ as  the canonical tomographic pulse set, and employed   greedy algorithm to search for smaller readout pulse set. It turns out that, unlike the previous example, it is hard  to   analytically construct  an optimal scheme for the current problem. Here we  resort to integer programming. Fig. \ref{result} shows our running results, which are listed together with  the methods of canonical tomography and greedy search. From the figure, we can see that for a 6 qubit-system, using the measurement scheme found by integer programming would save around 20\% of the experiment time compared with that if we use greedy method.   Moreover, integer programming allows   to confirm that the obtained solution is indeed optimal for system's size up to 5; see Table \ref{optimal}. These results clearly demonstrate the usefulness of the integer programming technique, that we could get  appreciable improvements   over what has been obtained before.


\begin{figure}[t]
\begin{center}
\includegraphics[width=0.85\linewidth]{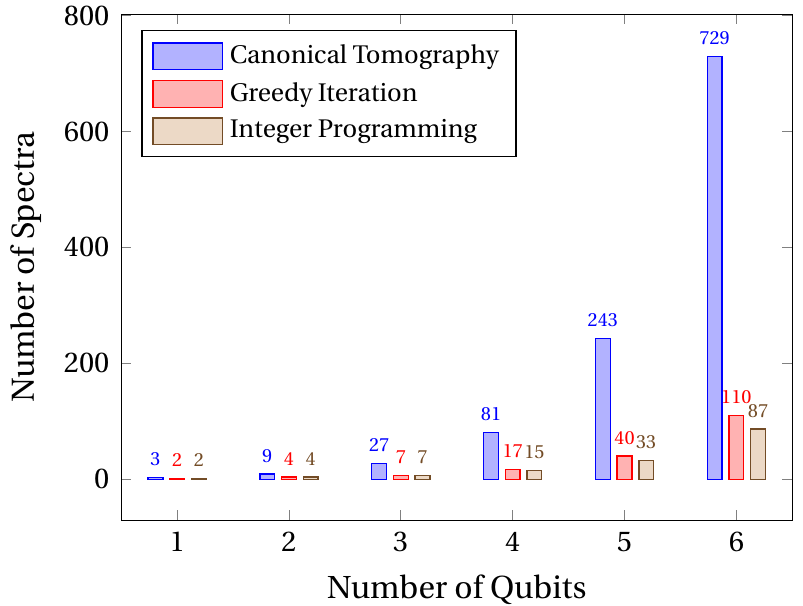}
\end{center}
\caption{Comparison of canonical tomography, results of iterative greedy algorithm in Ref. \cite{LM04}, and results of integer programming approach for complete tomography on homonuclear $n$-qubit systems with $n$ between 1 and 6.}
\label{result}
\end{figure}

\section{Summary}

Quantum state tomography plays an essential role in many quantum information processing experiments. Developing techniques that allow simplification of    state tomography experiments  is particularly pressing in regard of the situation that ever larger sized quantum devices are emerging in laboratories. In this work, we   have studied  the application of integer programming   to  the problem of  reducing     the number of required measurement settings  and the computational complexity of data processing.  The presented   test examples confirm
the usefulness of   the  integer programming approach. Our method can be easily incorporated into other existing tomographic strategies  \cite{Jackson13,NSPW10,Ma16}. Also, it is straightforward to generalize our results  to quantum process tomography experiments.
It is the hope that   integer programming formulation, as developed in this work, will become a useful tool in future tomographic experiments   for increasingly large quantum systems, overcoming the roadblock against further development in quantum technologies.

\emph{Acknowledgments}.  
Jun Li is supported by the National Basic Research Program of China (Grants No. 2014CB921403, No. 2016YFA0301201, No. 2014CB848700 and No. 2013CB921800), National Natural Science Foundation of China (Grants No. 11421063, No. 11534002, No. 11375167 and No. 11605005), the National Science Fund for Distinguished Young Scholars (Grant No. 11425523), and NSAF (Grant No. U1530401). Bei Zeng is supported by NSERC and CIFAR.

\end{document}